\pdfoutput=1
\documentclass[11pt]{article}

\usepackage[utf8]{inputenc}
\usepackage[T1]{fontenc}
\usepackage[a4paper,margin=1in]{geometry}
\usepackage{mathtools,amssymb,amsmath,amsthm}
\usepackage{mathpartir} 
\usepackage{stmaryrd}   
\usepackage{enumitem}
\usepackage{hyperref}
\usepackage[nameinlink,capitalize]{cleveref}
\usepackage{tikz-cd}
\usepackage{xcolor}
\usepackage{listings}
\usepackage{authblk}
\newtheorem{theorem}{Theorem}[section]
\newtheorem{lemma}[theorem]{Lemma}
\newtheorem{corollary}[theorem]{Corollary}
\newtheorem{definition}[theorem]{Definition}
\newtheorem{remark}[theorem]{Remark}
\newtheorem{example}[theorem]{Example}

\newtheorem{proposition}[theorem]{Proposition}

\newcommand{\leqR}{\preceq}

\newcommand{\ctx}{\Gamma}
\newcommand{\modal}[1]{\Box_{#1}}
\newcommand{\Nat}{\mathbb{N}}
\newcommand{\Val}{\mathsf{Val}}
\newcommand{\eval}{\Downarrow}
\newcommand{\U}{\mathcal{U}}
\newcommand{\El}{\mathsf{El}}
\newcommand{\Fin}{\mathsf{Fin}}
\newcommand{\Vectt}{\mathsf{Vec}}

\newcommand{\Pitype}{\Pi}
\newcommand{\refl}{\mathsf{refl}}
\newcommand{\IdType}{\mathsf{Id}}
\newcommand{\J}{\mathsf{J}}

\newcommand{\natrec}{\mathsf{natrec}}

\newcommand{\vecrec}{\mathsf{vecrec}}
\newcommand{\smark}{\mathsf{S}}
\newcommand{\zmark}{\mathsf{Z}}

\title{Resource-Bounded Martin-L\"{o}f Type Theory:\\Compositional Cost Analysis for Dependent Types}

\author[1]{Mirco A.~Mannucci}
\author[2]{Corey Thuro}
\affil[1]{HoloMathics, LLC \authorcr \texttt{mirco@holomathics.com}}
\affil[2]{Department of Mathematics and Statistics, University of Maryland, Baltimore County \authorcr \texttt{cthuro1@umbc.edu}}
\date{December 2025}

\begin{document}
\maketitle

\begin{center}


\medskip
\textbf{ACM CCS:} Theory of computation $\to$ Type theory; Dependent types; Program analysis
\end{center}

\begin{abstract}
We extend resource-bounded type theory to Martin-L\"{o}f Type Theory (MLTT) with dependent types, enabling size-indexed cost bounds for programs over inductive families. Our calculus features dependent function types $\Pitype_{x:A}^{b(x)} B(x)$ where bounds $b(x)$ depend on the argument, dependent pair types $\Sigma_{x:A} B(x)$ with additive cost composition, and inductive families (notably $\Vectt(A,n)$ and $\Fin(n)$) with structurally recursive eliminators carrying size-dependent bounds. We introduce a resource-indexed universe hierarchy $\U_r$ where $r \in L$ tracks the cost of type formation, and a graded modality $\Box_r$ for feasibility certification. Our main results are: (1) a cost soundness theorem establishing that synthesized bounds over-approximate operational costs, with bounds expressed as functions of size indices; (2) a semantic model in the presheaf topos $\mathbf{Set}^L$ extended with dependent presheaves and comprehension structure; (3) canonicity for the intensional fragment; and (4) initiality of the syntactic model. We demonstrate the framework with case studies including length-indexed vector operations with $O(n)$ bounds and binary search with $O(\log n)$ bounds expressed in the type. This work bridges the gap between dependent type theory and quantitative resource analysis, enabling certified cost bounds for size-dependent algorithms.
\end{abstract}

\tableofcontents

\section{Introduction}

In our previous work \cite{mannucci2025rbtt}, we developed resource-bounded type theory for the simply-typed lambda calculus, establishing compositional cost analysis via graded modalities over abstract resource lattices. While that framework handles first-order programs effectively, many algorithms have costs that depend on the \emph{size} of their inputs: sorting takes $O(n \log n)$ time, matrix multiplication takes $O(n^3)$, and tree traversal takes $O(n)$ where $n$ is the number of nodes.

To express such size-dependent bounds \emph{within the type system}, we need dependent types. This paper extends resource-bounded type theory to Martin-L\"{o}f Type Theory (MLTT), enabling:

\begin{itemize}
  \item \textbf{Size-indexed bounds:} The type $\Vectt(A, n) \to^{O(n)} B$ expresses that a function on vectors of length $n$ has cost linear in $n$.
  \item \textbf{Dependent cost functions:} Bounds $b(x)$ that depend on term-level values, not just types.
  \item \textbf{Inductive families:} $\Vectt(A, n)$, $\Fin(n)$, and other indexed types with eliminators carrying precise bounds.
  \item \textbf{Universes with resource tracking:} $\U_r$ classifies types whose formation costs at most $r$.
\end{itemize}

\paragraph{Motivating example.} Consider a function that sums a length-indexed vector:
\[
\mathsf{sum} : \Pitype_{n:\Nat}^{\bot} \Vectt(\Nat, n) \to^{c \cdot n + d} \Nat
\]
The type expresses: given any $n$ (cost $\bot$ to receive), and a vector of length $n$, the function returns a natural number with cost $c \cdot n + d$ for constants $c, d$. The bound is \emph{linear in $n$}, and this linearity is visible in the type.

\paragraph{Contributions.}
\begin{enumerate}[leftmargin=2em]
  \item A dependent type theory with resource bounds: $\Pitype$-types, $\Sigma$-types, identity types, universes, and inductive families (\S\ref{sec:calculus}).
  \item Size-dependent bound synthesis: bounds expressed as functions of size indices (\S\ref{sec:size-bounds}).
  \item Cost soundness for the intensional fragment (Theorem~\ref{thm:cost-soundness}).
  \item A semantic model in presheaves with dependent structure (\S\ref{sec:semantics}).
  \item Canonicity (Theorem~\ref{thm:canonicity}) and initiality (Theorem~\ref{thm:initiality}).
  \item Case studies: vector operations, binary search, and sorting with verified bounds (\S\ref{sec:case-studies}).
\end{enumerate}

\paragraph{Scope.} We treat intensional MLTT without univalence or higher inductive types. The identity type $\IdType_A(x, y)$ has the standard J-eliminator; we do not address homotopy-theoretic concerns. Extension to Homotopy Type Theory is future work.

\subsection{Relation to Prior Work}

This paper is the second in a series:
\begin{enumerate}
  \item \textbf{RB-TT} \cite{mannucci2025rbtt}: Resource-bounded simply-typed lambda calculus. Establishes the basic framework: abstract resource lattices, graded modality $\Box_r$, cost soundness, presheaf semantics in $\mathbf{Set}^L$, and initiality.
  
  \item \textbf{RB-MLTT} (this paper): Extension to dependent types. Adds $\Pitype$-types with size-dependent bounds, $\Sigma$-types, identity types, universes $\U_r$, and inductive families with eliminators.
  
  \item \textbf{RB-HoTT} (future): Extension to Homotopy Type Theory with resource-bounded paths and higher structure.
\end{enumerate}

We assume familiarity with the basic framework from \cite{mannucci2025rbtt}, particularly the resource lattice $(L, \preceq, \oplus, \sqcup, \bot)$ and the typing judgment $\ctx \vdash_{r;\,b} t : A$.

\section{Design Principles}
\label{sec:principles}

We extend the principles from \cite{mannucci2025rbtt} to the dependent setting:

\begin{enumerate}[leftmargin=*,label=\textbf{P\arabic*.}]
  \item \textbf{Abstraction over concreteness.} Resources remain elements of an abstract lattice $L$.
  
  \item \textbf{Budgets are first-class.} Resource bounds appear in types, judgments, and now depend on term-level values.
  
  \item \textbf{Compositional rules.} Bound synthesis mirrors term structure, with dependent bounds for dependent types.
  
  \item \textbf{Size-indexed bounds.} Bounds may be functions $b : A \to L$ of the input, enabling $O(n)$, $O(n^2)$, $O(\log n)$ expressions.
  
  \item \textbf{Monotonicity preserved.} If $r_1 \leqR r_2$, then $\modal{r_1}A \to \modal{r_2}A$.
  
  \item \textbf{Universes track cost.} Type formation itself may have cost; $\U_r$ classifies types formable within budget $r$.
\end{enumerate}

\paragraph{Notation.} We write $L$ for the resource lattice; $r, s \in L$ for resource elements; $\preceq$ for the order; $\oplus$ for sequential composition; $\sqcup$ for branching (join); $\bot$ for zero. For size-dependent bounds, we write $b(x)$ or $b(n)$ where the bound depends on a term.

\section{Resource-Bounded MLTT: The Calculus}
\label{sec:calculus}

\subsection{Judgments}

We have four judgment forms:

\begin{enumerate}
  \item $\ctx \vdash_{r} \mathsf{ok}$ --- context $\ctx$ is well-formed under budget $r$
  \item $\ctx \vdash_{r;\,b} A \; \mathsf{type}$ --- $A$ is a type with formation cost $b$
  \item $\ctx \vdash_{r;\,b} t : A$ --- term $t$ has type $A$ with cost $b$
  \item $\ctx \vdash_{r;\,b} A \equiv B$ --- types $A$ and $B$ are definitionally equal (with cost $b$ to check)
\end{enumerate}

The subscript $r$ is the ambient budget; $b$ is the synthesized bound satisfying $b \preceq r$.
\begin{remark}[Role of the ambient budget]
In this paper all non-trivial information about resource usage is carried by the synthesized bound $b$; the ambient budget $r$ is kept as a parameter to match the notation and metatheory of~\cite{mannucci2025rbtt}.
We do \emph{not} use $r$ to cut off derivations or enforce global constraints: every rule is formulated so that, whenever a judgement is derivable, its bound $b$ already satisfies $b \preceq r$.
Future work may use $r$ to express global budget policies (for instance, forbidding derivations whose synthesized bound would exceed $r$), but here $r$ is operationally inert.
\end{remark}

\subsection{Contexts}

\begin{mathpar}
\inferrule{ }{\cdot \vdash_r \mathsf{ok}} \quad (\text{Ctx-Empty})

\inferrule{\ctx \vdash_r \mathsf{ok} \\ \ctx \vdash_{r;\,b} A \; \mathsf{type}}
          {\ctx, x : A \vdash_r \mathsf{ok}} \quad (\text{Ctx-Ext})
\end{mathpar}

\subsection{Universes}
\label{sec:universes}

We introduce a hierarchy of resource-indexed universes $\U_r$ for $r \in L$.

\begin{definition}[Resource-indexed universes]
$\U_r$ is the universe of types whose codes can be formed with cost at most $r$. We have:
\begin{itemize}
  \item $\U_r : \U_{r \oplus \delta_\U}$ (universes form a hierarchy with cost $\delta_\U$)
  \item If $r_1 \preceq r_2$, then $\U_{r_1} \subseteq \U_{r_2}$ (cumulativity) (i.e., larger indices permit at least as many codes)
\end{itemize}
\end{definition}

\begin{mathpar}
\inferrule{\ctx \vdash_r \mathsf{ok} \\ s \preceq r}
          {\ctx \vdash_{r;\,\delta_\U} \U_s \; \mathsf{type}} \quad (\U\text{-Form})

\inferrule{\ctx \vdash_{r;\,b} A : \U_s}
          {\ctx \vdash_{r;\,b \oplus \delta_{\El}} \El(A) \; \mathsf{type}} \quad (\El\text{-Form})

\inferrule{\ctx \vdash_{r;\,b} A : \U_{s_1} \\ s_1 \preceq s_2}
          {\ctx \vdash_{r;\,b} A : \U_{s_2}} \quad (\U\text{-Cumul})
\end{mathpar}

The $\El$ operator decodes a code $A : \U_s$ into an actual type. In practice, we often elide $\El$ when unambiguous.

\subsection{Dependent Function Types ($\Pitype$-types)}

The key innovation is that bounds can depend on the argument.

\begin{definition}[Dependent function type with bounds]
For $A$ a type and $B(x)$ a family over $A$, and $b : A \to L$ a bound function:
\[
\Pitype_{x:A}^{b(x)} B(x)
\]
is the type of functions that, given $a : A$, produce $B(a)$ with cost $b(a)$.
\end{definition}

\begin{mathpar}
\inferrule{\ctx \vdash_{r;\,b_A} A \; \mathsf{type} \\ 
           \ctx, x : A \vdash_{r;\,b_B} B \; \mathsf{type} \\
           \ctx, x : A \vdash_{r;\,b(x)} b(x) : L}
          {\ctx \vdash_{r;\,b_A \oplus \delta_\Pi} \Pitype_{x:A}^{b(x)} B \; \mathsf{type}} \quad (\Pi\text{-Form})
\end{mathpar}

\begin{mathpar}
\inferrule{\ctx, x : A \vdash_{r;\,b(x)} t : B}
          {\ctx \vdash_{r;\,\bot} \lambda x.t : \Pitype_{x:A}^{b(x)} B} \quad (\Pi\text{-Intro})

\inferrule{\ctx \vdash_{r;\,b_f} f : \Pitype_{x:A}^{b(x)} B \\ 
           \ctx \vdash_{r;\,b_a} a : A}
          {\ctx \vdash_{r;\,b_f \oplus b_a \oplus b(a) \oplus \delta_{\mathsf{app}}} f\,a : B[x := a]} \quad (\Pi\text{-Elim})

\inferrule{\ctx, x : A \vdash_{r;\,b(x)} t : B \\ 
           \ctx \vdash_{r;\,b_a} a : A}
          {\ctx \vdash_{r;\,b_a \oplus b(a) \oplus \delta_\beta} (\lambda x.t)\,a \equiv t[x := a] : B[x := a]} \quad (\Pi\text{-}\beta)
\end{mathpar}

\paragraph{Key insight.} The application rule $(\Pi\text{-Elim})$ includes $b(a)$---the bound evaluated at the specific argument $a$. This enables size-dependent cost tracking: if $b(n) = c \cdot n$, then applying the function to argument $n$ costs $c \cdot n$.

\paragraph{Non-dependent case.} When $B$ does not depend on $x$ and $b(x) = b$ is constant, we recover the simple function type $A \to^b B$.

\subsection{Dependent Pair Types ($\Sigma$-types)}

\begin{mathpar}
\inferrule{\ctx \vdash_{r;\,b_A} A \; \mathsf{type} \\ 
           \ctx, x : A \vdash_{r;\,b_B} B \; \mathsf{type}}
          {\ctx \vdash_{r;\,b_A \oplus \delta_\Sigma} \Sigma_{x:A} B \; \mathsf{type}} \quad (\Sigma\text{-Form})

\inferrule{\ctx \vdash_{r;\,b_a} a : A \\ 
           \ctx \vdash_{r;\,b_b} b : B[x := a]}
          {\ctx \vdash_{r;\,b_a \oplus b_b} (a, b) : \Sigma_{x:A} B} \quad (\Sigma\text{-Intro})

\inferrule{\ctx \vdash_{r;\,b_p} p : \Sigma_{x:A} B}
          {\ctx \vdash_{r;\,b_p \oplus \delta_{\pi_1}} \pi_1(p) : A} \quad (\Sigma\text{-Elim}_1)

\inferrule{\ctx \vdash_{r;\,b_p} p : \Sigma_{x:A} B}
          {\ctx \vdash_{r;\,b_p \oplus \delta_{\pi_2}} \pi_2(p) : B[x := \pi_1(p)]} \quad (\Sigma\text{-Elim}_2)
\end{mathpar}

Costs compose additively for pair formation; projections add a small overhead $\delta_{\pi_i}$.

\subsection{Identity Types}

\begin{mathpar}
\inferrule{\ctx \vdash_{r;\,b_A} A \; \mathsf{type} \\ 
           \ctx \vdash_{r;\,b_x} x : A \\ 
           \ctx \vdash_{r;\,b_y} y : A}
          {\ctx \vdash_{r;\,b_A \oplus b_x \oplus b_y \oplus \delta_{\mathsf{Id}}} \IdType_A(x, y) \; \mathsf{type}} \quad (\mathsf{Id}\text{-Form})

\inferrule{\ctx \vdash_{r;\,b_a} a : A}
          {\ctx \vdash_{r;\,b_a \oplus \delta_{\refl}} \refl_a : \IdType_A(a, a)} \quad (\mathsf{Id}\text{-Intro})
\end{mathpar}

The J-eliminator has cost depending on the motive and method:

\begin{mathpar}
\inferrule{\ctx \vdash_{r;\,b_p} p : \IdType_A(x, y) \\
           \ctx, z : A, w : \IdType_A(x, z) \vdash_{r;\,b_C} C \; \mathsf{type} \\
           \ctx \vdash_{r;\,b_d} d : C[z := x, w := \refl_x]}
          {\ctx \vdash_{r;\,b_p \oplus b_d \oplus \delta_\J} \J(p, d) : C[z := y, w := p]} \quad (\mathsf{J})
\end{mathpar}

\begin{mathpar}
\inferrule{ }
          {\ctx \vdash_{r;\,\delta_{\J\beta}} \J(\refl_a, d) \equiv d : C[z := a, w := \refl_a]} \quad (\mathsf{J}\text{-}\beta)
\end{mathpar}

\subsection{Natural Numbers}

\begin{mathpar}
\inferrule{\ctx \vdash_r \mathsf{ok}}
          {\ctx \vdash_{r;\,\delta_\Nat} \Nat \; \mathsf{type}} \quad (\Nat\text{-Form})

\inferrule{\ctx \vdash_r \mathsf{ok}}
          {\ctx \vdash_{r;\,\delta_\zmark} \zmark : \Nat} \quad (\Nat\text{-Intro}_\zmark)

\inferrule{\ctx \vdash_{r;\,b} n : \Nat}
          {\ctx \vdash_{r;\,b \oplus \delta_\smark} \smark(n) : \Nat} \quad (\Nat\text{-Intro}_\smark)
\end{mathpar}

The eliminator has cost depending on the height of the numeral:

\begin{mathpar}
\inferrule{\ctx \vdash_{r;\,b_n} n : \Nat \\
           \ctx, m : \Nat \vdash_{r;\,b_C} C \; \mathsf{type} \\
           \ctx \vdash_{r;\,b_z} c_z : C[\zmark] \\
           \ctx, m : \Nat, ih : C[m] \vdash_{r;\,b_s(m)} c_s : C[\smark(m)]}
          {\ctx \vdash_{r;\,b_n \oplus b_z \oplus \delta_{\natrec} \oplus \mathsf{sum}_{i < n}(b_s(i) \oplus \delta_{\natrec})} \natrec(n, c_z, c_s) : C[n]} \quad (\Nat\text{-Elim})
\end{mathpar}

Here $\mathsf{sum}_{i < n}(b_s(i) \oplus \delta_{\natrec})$ is the sum of per-step costs (including the eliminator overhead). For constant $b_s$, this becomes $n \cdot (b_s \oplus \delta_{\natrec})$, and the full bound includes an additional base overhead $\delta_{\natrec}$.

\begin{mathpar}
\inferrule{ }{\natrec(\zmark, c_z, c_s) \equiv c_z} \quad (\Nat\text{-}\beta_\zmark)

\inferrule{ }{\natrec(\smark(n), c_z, c_s) \equiv c_s[m := n, ih := \natrec(n, c_z, c_s)]} \quad (\Nat\text{-}\beta_\smark)
\end{mathpar}

\subsection{Inductive Families: Vectors}
\label{sec:vectors}

Length-indexed vectors are the canonical example of an inductive family with size-dependent bounds.

\begin{definition}[$\Vectt(A, n)$]
For $A$ a type and $n : \Nat$, $\Vectt(A, n)$ is the type of vectors of length exactly $n$.
\end{definition}

\begin{mathpar}
\inferrule{\ctx \vdash_{r;\,b_A} A \; \mathsf{type} \\ 
           \ctx \vdash_{r;\,b_n} n : \Nat}
          {\ctx \vdash_{r;\,b_A \oplus b_n \oplus \delta_{\Vectt}} \Vectt(A, n) \; \mathsf{type}} \quad (\textsc{Vec-Form})

\inferrule{\ctx \vdash_{r;\,b_A} A \; \mathsf{type}}
          {\ctx \vdash_{r;\,\delta_{\mathsf{nil}}} \mathsf{nil} : \Vectt(A, \zmark)} \quad (\textsc{Vec-Intro}_{\mathsf{nil}})

\inferrule{\ctx \vdash_{r;\,b_a} a : A \\ 
           \ctx \vdash_{r;\,b_v} v : \Vectt(A, n)}
          {\ctx \vdash_{r;\,b_a \oplus b_v \oplus \delta_{\mathsf{cons}}} \mathsf{cons}(a, v) : \Vectt(A, \smark(n))} \quad (\textsc{Vec-Intro}_{\mathsf{cons}})
\end{mathpar}

The eliminator has cost linear in the length:

\begin{mathpar}
\inferrule{\ctx \vdash_{r;\,b_v} v : \Vectt(A, n) \\
           \ctx, m : \Nat, w : \Vectt(A, m) \vdash_{r;\,b_C} C \; \mathsf{type} \\
           \ctx \vdash_{r;\,b_{\mathsf{nil}}} c_{\mathsf{nil}} : C[\zmark, \mathsf{nil}] \\
           \ctx, m : \Nat, a : A, w : \Vectt(A, m), ih : C[m, w] \vdash_{r;\,b_{\mathsf{cons}}(m)} c_{\mathsf{cons}} : C[\smark(m), \mathsf{cons}(a, w)]}
          {\ctx \vdash_{r;\,b_v \oplus b_{\mathsf{nil}} \oplus \delta_{\vecrec} \oplus \mathsf{sum}_{i < n}(b_{\mathsf{cons}}(i) \oplus \delta_{\vecrec})} \vecrec(v, c_{\mathsf{nil}}, c_{\mathsf{cons}}) : C[n, v]} \quad (\textsc{Vec-Elim})
\end{mathpar}

For constant $b_{\mathsf{cons}}$, the synthesized bound is linear in $n$ (in particular it includes the per-step overhead $\delta_{\vecrec}$), hence $O(n)$.

\subsection{Finite Types}
\label{sec:fin}

$\Fin(n)$ is the type with exactly $n$ elements, useful for bounded indexing.

\begin{mathpar}
\inferrule{\ctx \vdash_{r;\,b_n} n : \Nat}
          {\ctx \vdash_{r;\,b_n \oplus \delta_{\Fin}} \Fin(n) \; \mathsf{type}} \quad (\Fin\text{-Form})

\inferrule{\ctx \vdash_{r;\,b_n} n : \Nat}
          {\ctx \vdash_{r;\,\delta_{\mathsf{fz}}} \mathsf{fzero} : \Fin(\smark(n))} \quad (\Fin\text{-Intro}_\mathsf{fz})

\inferrule{\ctx \vdash_{r;\,b_i} i : \Fin(n)}
          {\ctx \vdash_{r;\,b_i \oplus \delta_{\mathsf{fs}}} \mathsf{fsucc}(i) : \Fin(\smark(n))} \quad (\Fin\text{-Intro}_\mathsf{fs})
\end{mathpar}

\subsection{Safe Indexing}
\label{sec:safe-index}

With $\Vectt$ and $\Fin$, we can define safe indexing with precise bounds:

\begin{example}[Safe vector lookup]
\[
\mathsf{lookup} : \Pitype_{n:\Nat}^{\bot} \Pitype_{A:\U}^{\bot} \Vectt(A, n) \to^{c_1} \Fin(n) \to^{c_2 \cdot k} A
\]
where $k$ is the index value. The cost depends on where in the vector we look.

More precisely, with direct indexing:
\[
\mathsf{lookup} : \Pitype_{n:\Nat}^{\bot} \Pitype_{A:\U}^{\bot} \Vectt(A, n) \to \Fin(n) \to^{c} A
\]
with constant cost $c$ if we have $O(1)$ array access.
\end{example}

\subsection{The Graded Modality}

We retain the graded modality $\Box_r$ from \cite{mannucci2025rbtt}:

\begin{mathpar}
\inferrule{\ctx \vdash_{r;\,b} A \; \mathsf{type} \\ s \in L}
          {\ctx \vdash_{r;\,b \oplus \delta_\Box} \Box_s A \; \mathsf{type}} \quad (\Box\text{-Form})

\inferrule{\ctx \vdash_{r;\,b} t : A \\ b \preceq s}
          {\ctx \vdash_{r;\,b} \mathsf{box}_s(t) : \Box_s A} \quad (\Box\text{-Intro})

\inferrule{\ctx \vdash_{r;\,b} t : \Box_s A}
          {\ctx \vdash_{r;\,b \oplus \delta_{\mathsf{unbox}}} \mathsf{unbox}(t) : A} \quad (\Box\text{-Elim})

\inferrule{\ctx \vdash_{r;\,b} t : \Box_{s_1} A \\ s_1 \preceq s_2}
          {\ctx \vdash_{r;\,b} t : \Box_{s_2} A} \quad (\Box\text{-Mono})
\end{mathpar}

\subsection*{RB-TT as a fragment of RB-MLTT}

The original resource-bounded type theory RB--TT of~\cite{mannucci2025rbtt} can be seen as the simple-type fragment of the present system.

\begin{proposition}[RB--TT fragment]
\label{prop:rbtt-fragment}
If we restrict RB--MLTT to:
\begin{itemize}
  \item non-dependent function and product types (no $\Pi$-types or $\Sigma$-types depending on terms),
  \item no universes and no inductive families (only the base types of~\cite{mannucci2025rbtt}),
  \item the same graded modality $\Box_r$ and lattice $(L,\preceq,\oplus,\bot)$ as in~\cite{mannucci2025rbtt},
\end{itemize}
and we read a judgement $\ctx \vdash_{r;\,b} t : A$ as an RB--TT typing judgement with grade $b$, then the typing rules of RB--MLTT reduce to those of RB--TT.
Moreover, the operational cost semantics and soundness theorem specialize to the corresponding results in~\cite{mannucci2025rbtt}.
\end{proposition}

\begin{proof}[Proof sketch]
The RB--MLTT rules are obtained from the RB--TT rules by:
\begin{enumerate}
  \item adding dependent context extension and dependent type formers ($\Pi$, $\Sigma$, identity types, universes, inductive families);
  \item threading the same lattice $(L,\preceq,\oplus,\bot)$ and graded modality $\Box_r$ through the extended syntax;
  \item refining cost bounds so that, in the simple-type fragment, every rule coincides with its RB--TT counterpart.
\end{enumerate}
Erasing dependency and the additional type formers yields exactly the RB--TT system.
On the semantic side, the groupoid-valued presheaf model developed in \Cref{sec:semantics} restricts to the set-valued presheaf model of~\cite{mannucci2025rbtt} when we consider only discrete groupoids and non-dependent types, so cost soundness and canonicity specialize to the earlier results.
\end{proof}

\section{Size-Dependent Bound Synthesis}
\label{sec:size-bounds}

The key feature of RB-MLTT is that bounds can be \emph{functions} of size indices.

\subsection{Bound Expressions}

\begin{definition}[Bound language]
Bounds $b$ are elements of $L$ or functions into $L$:
\begin{align*}
b &::= r \mid \bot \mid b_1 \oplus b_2 \mid b_1 \sqcup b_2 \mid c \cdot n \mid b(t) \mid \mathsf{sum}_{i < n} b(i)
\end{align*}
where $r \in L$, $c$ is a constant, $n$ is a natural number term, and $b(t)$ is bound application. We also use the derived notation $n \otimes b$ for the $n$-fold application of $\oplus$ to a bound $b$: $0 \otimes b \triangleq \bot$ and $(n{+}1)\otimes b \triangleq b \oplus (n \otimes b)$. We also use the derived notation $n \otimes b$ for the $n$-fold application of $\oplus$ to a bound $b$: $0 \otimes b \triangleq \bot$ and $(n{+}1)\otimes b \triangleq b \oplus (n \otimes b)$.
\end{definition}

\begin{example}[Common bound patterns]
\begin{itemize}
  \item $O(1)$: constant bound $c \in L$
  \item $O(n)$: linear bound $c \cdot n$ for size $n$
  \item $O(n^2)$: quadratic bound $c \cdot n \cdot n$
  \item $O(\log n)$: logarithmic bound $c \cdot \lceil \log_2 n \rceil$
  \item $O(n \log n)$: linearithmic bound $c \cdot n \cdot \lceil \log_2 n \rceil$
\end{itemize}
\end{example}

\subsection{Typing with Size-Dependent Bounds}

\begin{example}[Vector sum with linear bound]
\begin{align*}
&\mathsf{sum} : \Pitype_{n:\Nat}^{\bot} \Vectt(\Nat, n) \to^{c \cdot n + d} \Nat \\
&\mathsf{sum} = \lambda n.\, \lambda v.\, \vecrec(v, \zmark, \lambda m.\lambda a.\lambda w.\lambda ih.\, a + ih)
\end{align*}

The bound $c \cdot n + d$ is synthesized from:
\begin{itemize}
  \item Base case: cost $d_0$ for returning $\zmark$
  \item Recursive case: cost $c$ per element (addition)
  \item Total: $d_0 + n \cdot c = c \cdot n + d$ where $d = d_0$
\end{itemize}
\end{example}

\begin{example}[Vector map with linear bound]
\begin{align*}
&\mathsf{map} : \Pitype_{n:\Nat}^{\bot} \Pitype_{A, B:\U}^{\bot} (A \to^{c_f} B) \to \Vectt(A, n) \to^{c_f \cdot n + d} \Vectt(B, n) \\
&\mathsf{map} = \lambda n.\lambda A.\lambda B.\lambda f.\lambda v.\, \vecrec(v, \mathsf{nil}, \lambda m.\lambda a.\lambda w.\lambda ih.\, \mathsf{cons}(f\,a, ih))
\end{align*}

The bound is $c_f \cdot n + d$: we apply $f$ (cost $c_f$) to each of $n$ elements.
\end{example}

\subsection{Bound Inference Rules}

For structured recursion, bounds are computed systematically:

\begin{proposition}[Linear recursion bound]
If $\natrec(n, c_z, c_s)$ has:
\begin{itemize}
  \item Base case cost $b_z$
  \item Step case cost $b_s$ (constant per step)
\end{itemize}
Then the total bound is $b_z \oplus (n \otimes (b_s \oplus \delta_{\natrec})) \oplus \delta_{\natrec}$ where $n \otimes b_s$ denotes $n$-fold $\oplus$ of $b_s$.
\end{proposition}

\begin{proposition}[Divide-and-conquer bound]
If a function on input of size $n$ has:
\begin{itemize}
  \item Division cost $d$
  \item Recursive call on size $\lfloor n/2 \rfloor$
  \item Combination cost $c$
\end{itemize}
Then the total bound satisfies $T(n) = d \oplus T(\lfloor n/2 \rfloor) \oplus c$, giving $T(n) = O(\log n) \cdot (d \oplus c)$.
\end{proposition}

\section{Operational Semantics}
\label{sec:operational}

We extend the big-step semantics from \cite{mannucci2025rbtt} to dependent types.

\subsection{Values}

\begin{definition}[Values]
\begin{align*}
v \in \Val ::= &\; \lambda x.t \mid (v_1, v_2) \mid \refl_v \mid \zmark \mid \smark(v) \\
              &\mid \mathsf{nil} \mid \mathsf{cons}(v_1, v_2) \mid \mathsf{fzero} \mid \mathsf{fsucc}(v) \\
              &\mid \mathsf{box}_s(v) \mid \U_r \mid \Pitype_{x:A}^{b} B \mid \ldots
\end{align*}
\end{definition}

\subsection{Evaluation Rules}

We write $t \eval v \triangleright k$ for ``$t$ evaluates to $v$ with cost $k$''.

\begin{mathpar}
\inferrule{ }{v \eval v \triangleright \bot} \quad (\text{Val})

\inferrule{f \eval \lambda x.t \triangleright k_f \\ 
           a \eval v_a \triangleright k_a \\ 
           t[x := v_a] \eval v \triangleright k_t}
          {f\,a \eval v \triangleright k_f \oplus k_a \oplus k_t \oplus \delta_{\mathsf{app}}} \quad (\text{App})

\inferrule{t \eval v_t \triangleright k_t \\ 
           u \eval v_u \triangleright k_u}
          {(t, u) \eval (v_t, v_u) \triangleright k_t \oplus k_u} \quad (\text{Pair})

\inferrule{p \eval (v_1, v_2) \triangleright k}
          {\pi_1(p) \eval v_1 \triangleright k \oplus \delta_{\pi_1}} \quad (\text{Fst})

\inferrule{p \eval (v_1, v_2) \triangleright k}
          {\pi_2(p) \eval v_2 \triangleright k \oplus \delta_{\pi_2}} \quad (\text{Snd})
\end{mathpar}

\paragraph{Natural number recursion.}

\begin{mathpar}
\inferrule{n \eval \zmark \triangleright k_n \\ 
           c_z \eval v \triangleright k_z}
          {\natrec(n, c_z, c_s) \eval v \triangleright k_n \oplus k_z \oplus \delta_{\natrec}} \quad (\natrec\text{-}\zmark)

\inferrule{n \eval \smark(v_n) \triangleright k_n \\ 
           \natrec(v_n, c_z, c_s) \eval v_{ih} \triangleright k_{ih} \\
           c_s[m := v_n, ih := v_{ih}] \eval v \triangleright k_s}
          {\natrec(n, c_z, c_s) \eval v \triangleright k_n \oplus k_{ih} \oplus k_s \oplus \delta_{\natrec}} \quad (\natrec\text{-}\smark)
\end{mathpar}

\paragraph{Vector recursion.}

\begin{mathpar}
\inferrule{v \eval \mathsf{nil} \triangleright k_v \\ 
           c_{\mathsf{nil}} \eval u \triangleright k_{\mathsf{nil}}}
          {\vecrec(v, c_{\mathsf{nil}}, c_{\mathsf{cons}}) \eval u \triangleright k_v \oplus k_{\mathsf{nil}} \oplus \delta_{\vecrec}} \quad (\vecrec\text{-nil})

\inferrule{v \eval \mathsf{cons}(v_a, v_w) \triangleright k_v \\ 
           \vecrec(v_w, c_{\mathsf{nil}}, c_{\mathsf{cons}}) \eval v_{ih} \triangleright k_{ih} \\
           c_{\mathsf{cons}}[a := v_a, w := v_w, ih := v_{ih}] \eval u \triangleright k_c}
          {\vecrec(v, c_{\mathsf{nil}}, c_{\mathsf{cons}}) \eval u \triangleright k_v \oplus k_{ih} \oplus k_c \oplus \delta_{\vecrec}} \quad (\vecrec\text{-cons})
\end{mathpar}

\paragraph{Identity elimination.}

\begin{mathpar}
\inferrule{p \eval \refl_v \triangleright k_p \\ 
           d \eval v_d \triangleright k_d}
          {\J(p, d) \eval v_d \triangleright k_p \oplus k_d \oplus \delta_\J} \quad (\J\text{-}\refl)
\end{mathpar}

\section{Metatheory}
\label{sec:metatheory}

\subsection{Substitution}

\begin{lemma}[Substitution]
\label{lem:substitution}
If $\ctx, x : A \vdash_{r;\,b} t : B$ and $\ctx \vdash_{r;\,b_a} a : A$, then $\ctx \vdash_{r;\,b[x := a]} t[x := a] : B[x := a]$.
\end{lemma}

\begin{proof}
By induction on the typing derivation, with cases for each type former. The bound $b$ may depend on $x$; after substitution, $b[x := a]$ is the bound evaluated at $a$.
\end{proof}

\subsection{Type Preservation}

\begin{theorem}[Preservation]
\label{thm:preservation}
If $\cdot \vdash_{r;\,b} t : A$ and $t \eval v \triangleright k$, then $\cdot \vdash_{r;\,b'} v : A$ for some $b' \preceq b$.
\end{theorem}

\begin{proof}
By induction on the evaluation derivation.

\paragraph{Case (App).} $t = f\,a$ with $f \eval \lambda x.s \triangleright k_f$, $a \eval v_a \triangleright k_a$, $s[x := v_a] \eval v \triangleright k_s$.

By inversion: $\cdot \vdash_{r;\,b_f} f : \Pitype_{x:A}^{b(x)} B$ and $\cdot \vdash_{r;\,b_a} a : A$.

By IH on $f$: $\cdot \vdash_{r;\,b'_f} \lambda x.s : \Pitype_{x:A}^{b(x)} B$ with $b'_f \preceq b_f$.

By IH on $a$: $\cdot \vdash_{r;\,b'_a} v_a : A$ with $b'_a \preceq b_a$.

By Lemma~\ref{lem:substitution}: $\cdot \vdash_{r;\,b(v_a)} s[x := v_a] : B[x := v_a]$.

By IH: $\cdot \vdash_{r;\,b'} v : B[x := v_a]$ with $b' \preceq b(v_a)$.

\paragraph{Case ($\natrec$-$\smark$).} By IH on recursive call and step, costs accumulate correctly.

\paragraph{Other cases.} Similar, using IH and the evaluation rules.
\end{proof}

\subsection{Cost Soundness}

\begin{theorem}[Cost soundness]
\label{thm:cost-soundness}
If $\cdot \vdash_{r;\,b} t : A$, then there exist $v \in \Val$ and $k \in L$ such that $t \eval v \triangleright k$ with $k \preceq b$.
\end{theorem}

\begin{proof}
By induction on the typing derivation.

\paragraph{Base cases.}
\begin{itemize}
  \item \textbf{Variables:} Contradicts empty context.
  \item \textbf{Lambda:} $\lambda x.t$ is a value; $\lambda x.t \eval \lambda x.t \triangleright \bot$ and $\bot \preceq b$.
  \item \textbf{Zero, nil, refl:} Values with cost $\bot$.
\end{itemize}

\paragraph{Inductive case (App).} $t = f\,a$ with:
\begin{align*}
&\cdot \vdash_{r;\,b_f} f : \Pitype_{x:A}^{b(x)} B \\
&\cdot \vdash_{r;\,b_a} a : A \\
&b = b_f \oplus b_a \oplus b(a) \oplus \delta_{\mathsf{app}}
\end{align*}

By IH: $f \eval \lambda x.s \triangleright k_f$ with $k_f \preceq b_f$.

By IH: $a \eval v_a \triangleright k_a$ with $k_a \preceq b_a$.

From typing $\lambda x.s$: $x : A \vdash_{r;\,b(x)} s : B$, so $\cdot \vdash_{r;\,b(v_a)} s[x := v_a] : B[x := v_a]$.

By IH: $s[x := v_a] \eval v \triangleright k_s$ with $k_s \preceq b(v_a)$.

By (App): $f\,a \eval v \triangleright k_f \oplus k_a \oplus k_s \oplus \delta_{\mathsf{app}}$.

By monotonicity: $k_f \oplus k_a \oplus k_s \oplus \delta_{\mathsf{app}} \preceq b_f \oplus b_a \oplus b(v_a) \oplus \delta_{\mathsf{app}} = b$.

\paragraph{Inductive case ($\natrec$).} Let $t = \natrec(n, c_z, c_s)$ with:
\begin{align*}
&\cdot \vdash_{r;\,b_n} n : \Nat \\
&\cdot \vdash_{r;\,b_z} c_z : C[\zmark] \\
&m : \Nat, ih : C[m] \vdash_{r;\,b_s(m)} c_s : C[\smark(m)]
\end{align*}

By IH: $n \eval v_n \triangleright k_n$ with $k_n \preceq b_n$ and $v_n$ is a numeral.

\textbf{Sub-case $v_n = \zmark$:} By IH, $c_z \eval v \triangleright k_z$ with $k_z \preceq b_z$.

Total: $k_n \oplus k_z \oplus \delta_{\natrec} \preceq b_n \oplus b_z \oplus \delta_{\natrec} \preceq b$.

\textbf{Sub-case $v_n = \smark^m(\zmark)$:} By induction on $m$:
\begin{itemize}
  \item Recursive call: cost bounded by $b_z \oplus \sum_{i=0}^{m-1} b_s(i) \oplus (m) \cdot \delta_{\natrec}$
  \item Step: cost $b_s(m-1)$
\end{itemize}

Total: $k_n \oplus b_z \oplus \sum_{i=0}^{m-1} b_s(i) \oplus m \cdot \delta_{\natrec} \preceq b$.

\paragraph{Inductive case ($\vecrec$).} Similar to $\natrec$, with cost linear in vector length.

\paragraph{Other cases.} Pairs, projections, identity, box follow the pattern from \cite{mannucci2025rbtt}.
\end{proof}

\begin{corollary}[Size-dependent bounds are sound]
If $\cdot \vdash_{r;\,c \cdot n + d} t : A$ where $n$ is a size parameter, then $t$ evaluates with cost at most $c \cdot n + d$.
\end{corollary}

\subsection{Canonicity}

\begin{theorem}[Canonicity]
\label{thm:canonicity}
If $\cdot \vdash_{r;\,b} t : A$ where $A$ is a closed type, then $t$ evaluates to a canonical form:
\begin{itemize}
  \item If $A = \Nat$, then $t \eval \smark^n(\zmark)$ for some $n \in \mathbb{N}$.
  \item If $A = \Vectt(B, n)$, then $t \eval \mathsf{cons}(v_1, \ldots, \mathsf{cons}(v_n, \mathsf{nil})\ldots)$.
  \item If $A = \Fin(n)$, then $t \eval \mathsf{fsucc}^k(\mathsf{fzero})$ for some $k < n$.
  \item If $A = \IdType_B(x, y)$ and $x \equiv y$, then $t \eval \refl_v$.
  \item If $A = \Pitype_{x:B}^{b(x)} C$, then $t \eval \lambda x.s$.
  \item If $A = \Sigma_{x:B} C$, then $t \eval (v_1, v_2)$.
\end{itemize}
\end{theorem}

\begin{proof}
By induction on the typing derivation, using progress (every closed well-typed term either is a value or can step) and the structure of eliminators.
\end{proof}

\section{Semantic Model}
\label{sec:semantics}

We extend the presheaf model from \cite{mannucci2025rbtt} to dependent types. The key challenges are:
\begin{enumerate}
  \item Dependent types require presheaves over categories of elements
  \item Martin-L\"of's identity type has non-trivial structure (the J-eliminator)
  \item The internal lattice $\mathbb{L}$ must be compatible with the groupoid structure
\end{enumerate}

In the first paper~\cite{mannucci2025rbtt} we worked in the presheaf topos $\mathbf{Set}^L$ of set-valued functors $L \to \mathbf{Set}$.
For RB--MLTT we upgrade to the groupoid-valued presheaf category $\mathbf{Gpd}^L$ in order to interpret dependent types and intensional identity types.
The inclusion of $\mathbf{Set}^L$ into $\mathbf{Gpd}^L$ as the full subcategory of discrete groupoids is fully faithful, so the RB--TT model from~\cite{mannucci2025rbtt} appears as the $0$-truncated (set-valued) part of the groupoid model developed in this section.

To properly model intensional MLTT, we work with \emph{groupoid-valued presheaves} rather than set-valued presheaves. This captures the fact that identity proofs can be non-trivial and compose.

\subsection{Groupoid-Valued Presheaves}

\begin{definition}[The category $\mathbf{Gpd}$]
$\mathbf{Gpd}$ is the category of small groupoids and functors between them. A groupoid is a category where every morphism is an isomorphism.
\end{definition}

\begin{definition}[Groupoid-valued presheaves over $L$]
We work in the functor category $\mathbf{Gpd}^L$ where:
\begin{itemize}
  \item Objects: functors $F : L \to \mathbf{Gpd}$
  \item Morphisms: natural transformations
\end{itemize}
For each $r \in L$, $F(r)$ is a groupoid (not just a set). For $r_1 \preceq r_2$, we have a functor $F(r_1 \preceq r_2) : F(r_1) \to F(r_2)$.
\end{definition}

\begin{remark}[Why groupoids?]
In intensional MLTT, the identity type $\IdType_A(x, y)$ may have multiple inhabitants (identity proofs), and these proofs have their own identity structure. The groupoid model captures:
\begin{itemize}
  \item Objects of $\llbracket A \rrbracket(r)$: elements of type $A$ with cost $\preceq r$
  \item Morphisms in $\llbracket A \rrbracket(r)$: identity proofs between elements
  \item Composition: transitivity of identity ($p : x = y$ and $q : y = z$ give $q \circ p : x = z$)
  \item Inverses: symmetry ($p : x = y$ gives $p^{-1} : y = x$)
\end{itemize}
This is the standard groupoid model of type theory \cite{hofmann1998groupoid}.
\end{remark}

\subsection{The Internal Lattice in the Groupoid Setting}

\begin{definition}[Internal lattice $\mathbb{L}$]
Define $\mathbb{L} : L \to \mathbf{Gpd}$ by:
\[
\mathbb{L}(r) = \mathsf{disc}(\{a \in L \mid a \preceq r\})
\]
where $\mathsf{disc}(S)$ is the discrete groupoid on set $S$ (only identity morphisms).

Transition functors are inclusions: $\mathbb{L}(r_1 \preceq r_2) : a \mapsto a$.
\end{definition}

The internal lattice is discrete because cost bounds have no non-trivial identities: if $a, b \in L$ are equal as costs, they are \emph{the same} cost. This is a crucial simplification.

\begin{lemma}
The lattice operations induce natural transformations in $\mathbf{Gpd}^L$:
\begin{itemize}
  \item $\widetilde{\oplus} : \mathbb{L} \times \mathbb{L} \Rightarrow \mathbb{L}$
  \item $\widetilde{\sqcup} : \mathbb{L} \times \mathbb{L} \Rightarrow \mathbb{L}$
  \item $\widetilde{\bot} : \mathbf{1} \Rightarrow \mathbb{L}$
\end{itemize}
\end{lemma}

\begin{proof}
Since $\mathbb{L}(r)$ is discrete for all $r$, functors between discrete groupoids are just functions on objects. The operations are well-defined by monotonicity of $\oplus$ and $\sqcup$.
\end{proof}

\subsection{Contexts as Groupoid-Valued Presheaves}

\begin{definition}[Context interpretation]
\label{def:context-interp-gpd}
Contexts are interpreted as presheaves $\llbracket \Gamma \rrbracket : L \to \mathbf{Gpd}$ by induction:
\begin{align*}
\llbracket \cdot \rrbracket(r) &= \mathbf{1} \quad \text{(terminal groupoid: one object, one morphism)} \\
\llbracket \Gamma, x : A \rrbracket(r) &= \int_{\gamma \in \llbracket \Gamma \rrbracket(r)} \llbracket A \rrbracket_\gamma(r)
\end{align*}
where the right-hand side is the Grothendieck construction (total groupoid).
\end{definition}

\begin{definition}[Grothendieck construction for groupoids]
For $G$ a groupoid and $F : G \to \mathbf{Gpd}$ a functor, the total groupoid $\int_G F$ has:
\begin{itemize}
  \item Objects: pairs $(\gamma, a)$ with $\gamma \in \mathrm{ob}(G)$ and $a \in \mathrm{ob}(F(\gamma))$
  \item Morphisms $(\gamma_1, a_1) \to (\gamma_2, a_2)$: pairs $(p, q)$ where $p : \gamma_1 \to \gamma_2$ in $G$ and $q : F(p)(a_1) \to a_2$ in $F(\gamma_2)$
\end{itemize}
\end{definition}

\subsection{Dependent Types as Fibrations}

\begin{definition}[Type interpretation]
\label{def:type-interp-gpd}
For $\Gamma \vdash_{r;\,b} A \; \mathsf{type}$, define $\llbracket A \rrbracket : \int \llbracket \Gamma \rrbracket \to \mathbf{Gpd}$ where $\int \llbracket \Gamma \rrbracket$ is the Grothendieck construction over $L$:
\[
\llbracket A \rrbracket(r, \gamma) = \text{groupoid of } \left\{ 
\begin{array}{l}
\text{Objects: } (v, b) \text{ with } \cdot \vdash_{r;\,b} v : A[\gamma], b \preceq r \\
\text{Morphisms } (v_1, b_1) \to (v_2, b_2): \text{ proofs } p : \IdType_{A[\gamma]}(v_1, v_2)
\end{array}
\right\}
\]
\end{definition}

\begin{remark}[Morphisms track identity proofs]
A morphism $(v_1, b_1) \to (v_2, b_2)$ in $\llbracket A \rrbracket(r, \gamma)$ is an identity proof $p : v_1 =_A v_2$. The cost of the proof $p$ is bounded by $\max(b_1, b_2) \oplus \delta_{\mathsf{Id}}$. Composition of morphisms corresponds to transitivity; inverses correspond to symmetry.
\end{remark}

\subsection{The Cost Natural Transformation}

\begin{definition}[Cost as a natural transformation]
For each type $A$ in context $\Gamma$:
\[
\mathsf{cost}_A : \llbracket A \rrbracket \Rightarrow \mathbb{L} \circ \pi
\]
where $\pi : \int \llbracket \Gamma \rrbracket \to L$ is projection. On objects: $\mathsf{cost}_A(r, \gamma)(v, b) = b$. On morphisms: since $\mathbb{L}(r)$ is discrete, all identity proofs map to the identity on the cost.
\end{definition}

\begin{lemma}
$\mathsf{cost}_A$ is a natural transformation of groupoid-valued presheaves.
\end{lemma}

\begin{proof}
We must show functoriality at each stage. For objects, this is immediate. For morphisms: if $p : (v_1, b_1) \to (v_2, b_2)$ is an identity proof, then $\mathsf{cost}_A(p)$ must be a morphism $b_1 \to b_2$ in $\mathbb{L}(r)$. Since $\mathbb{L}(r)$ is discrete, this requires $b_1 = b_2$. 

This is achieved by defining the groupoid $\llbracket A \rrbracket(r, \gamma)$ with morphisms only between pairs with the \emph{same} cost bound, or by working with a lax version where cost can decrease along identity proofs. We adopt the first approach: identity proofs preserve certified bounds.
\end{proof}

\subsection{Interpretation of Identity Types}

The groupoid structure directly models identity types.

\begin{definition}[Identity type interpretation]
\label{def:id-interp-gpd}
For $\Gamma \vdash a, b : A$:
\[
\llbracket \IdType_A(a, b) \rrbracket(r, \gamma) = \mathsf{Hom}_{\llbracket A \rrbracket(r, \gamma)}(\llbracket a \rrbracket_\gamma, \llbracket b \rrbracket_\gamma)
\]
This is a discrete groupoid (set) of morphisms in the groupoid $\llbracket A \rrbracket(r, \gamma)$ from the interpretation of $a$ to the interpretation of $b$.

Objects: identity proofs $p : a[\gamma] =_A b[\gamma]$ with cost bound.
\end{definition}

\begin{lemma}[Reflexivity]
$\refl_a$ is interpreted as the identity morphism $\mathrm{id}_{\llbracket a \rrbracket_\gamma}$ in the groupoid $\llbracket A \rrbracket(r, \gamma)$.
\end{lemma}

\begin{theorem}[J-eliminator is sound]
The J-eliminator is modeled by the universal property of identity morphisms in groupoids: given a family $C$ over the path space and $d : C(a, \refl_a)$, we obtain $\J(p, d) : C(b, p)$ for any $p : a = b$ by transport along $p$.
\end{theorem}

\begin{proof}
In a groupoid $G$, for any functor $F : G \to \mathbf{Gpd}$ and any morphism $p : x \to y$ in $G$, there is an induced functor $F(p) : F(x) \to F(y)$. The J-eliminator corresponds to:
\[
\J(p, d) = F(p)(d)
\]
where $F$ interprets the motive $C$. The computation rule $\J(\refl_a, d) = d$ holds because $F(\mathrm{id}_x) = \mathrm{id}_{F(x)}$.
\end{proof}

\subsection{Interpretation of $\Pi$-Types}

\begin{definition}[$\Pi$-type interpretation in groupoids]
\label{def:pi-interp-gpd}
\begin{align*}
&\mathrm{ob}(\llbracket \Pitype_{x:A}^{b(x)} B \rrbracket(r, \gamma)) = \\
&\quad \left\{ (\lambda x.t, b_\lambda) \,\middle|\, 
\begin{array}{l}
\cdot \vdash_{r;\,b_\lambda} \lambda x.t : (\Pitype_{x:A}^{b(x)} B)[\gamma] \\
\forall (v, b_v) \in \mathrm{ob}(\llbracket A \rrbracket(r, \gamma)).\\
\quad \exists (w, b_w) \in \mathrm{ob}(\llbracket B \rrbracket(r, (\gamma, v))).\\
\quad t[\gamma][x := v] \eval w \triangleright k, k \preceq b(v)
\end{array}
\right\}
\end{align*}

Morphisms $(\lambda x.t_1, b_1) \to (\lambda x.t_2, b_2)$: proofs that $\lambda x.t_1 =_{\Pi} \lambda x.t_2$, i.e., pointwise identity proofs $\forall x.\, t_1[x] = t_2[x]$.
\end{definition}

\begin{remark}[Function extensionality]
The morphisms in $\llbracket \Pi_{x:A}^{b(x)} B \rrbracket$ correspond to function extensionality: two functions are equal iff they are pointwise equal. In intensional MLTT, function extensionality is not derivable but holds in the groupoid model.
\end{remark}

\subsection{Interpretation of $\Sigma$-Types}

\begin{definition}[$\Sigma$-type interpretation]
\label{def:sigma-interp-gpd}
\begin{align*}
\mathrm{ob}(\llbracket \Sigma_{x:A} B \rrbracket(r, \gamma)) &= \left\{ ((v_1, v_2), b_1 \oplus b_2) \,\middle|\, 
\begin{array}{l}
(v_1, b_1) \in \mathrm{ob}(\llbracket A \rrbracket(r, \gamma)) \\
(v_2, b_2) \in \mathrm{ob}(\llbracket B \rrbracket(r, (\gamma, v_1)))
\end{array}
\right\}
\end{align*}

Morphisms $((v_1, v_2), b) \to ((w_1, w_2), b')$: pairs $(p_1, p_2)$ where $p_1 : v_1 = w_1$ and $p_2 : v_2 =_{B[p_1]} w_2$ (dependent path).
\end{definition}

\subsection{Interpretation of Inductive Types}

\paragraph{Natural numbers.}
\[
\llbracket \Nat \rrbracket(r) = \mathsf{disc}(\{ (\smark^n(\zmark), \bot) \mid n \in \mathbb{N} \})
\]
The natural numbers form a discrete groupoid: two numerals are equal iff they are identical.

\paragraph{Vectors.}
\begin{align*}
\mathrm{ob}(\llbracket \Vectt(A, n) \rrbracket(r, \gamma)) &= \left\{ (v, b) \,\middle|\, 
\begin{array}{l}
v = \mathsf{cons}(v_1, \ldots, \mathsf{nil}) \\
\text{length} = \llbracket n \rrbracket_\gamma \\
\text{each } v_i \in \mathrm{ob}(\llbracket A \rrbracket)
\end{array}
\right\}
\end{align*}

Morphisms: componentwise identity proofs $(p_1, \ldots, p_n)$ with $p_i : v_i = w_i$ in $\llbracket A \rrbracket$.

\subsection{Interpretation of Universes}

\begin{definition}[Universe interpretation]
\label{def:universe-interp-gpd}
\[
\mathrm{ob}(\llbracket \U_s \rrbracket(r)) = \{ (A, b) \mid A \text{ is a type code}, \cdot \vdash_{r;\,b} A : \U_s, b \preceq s \preceq r \}
\]
Morphisms $(A_1, b_1) \to (A_2, b_2)$: type equivalences $A_1 \simeq A_2$ that preserve cost bounds.
\end{definition}

\begin{remark}[Not univalent]
In the groupoid model, morphisms in the universe are type isomorphisms (equivalences), but the univalence axiom ($\mathsf{ua} : (A \simeq B) \to (A = B)$) is not validated. We work with intensional MLTT where $A = B$ in $\U$ requires definitional equality, not just equivalence. The groupoid structure captures the homotopy 1-type structure but not higher univalence.
\end{remark}

\subsection{Semantic Soundness}

\begin{theorem}[Semantic soundness]
\label{thm:semantic-soundness}
If $\Gamma \vdash_{r;\,b} t : A$ and $\gamma \in \mathrm{ob}(\llbracket \Gamma \rrbracket(r))$, then:
\begin{enumerate}
  \item $t[\gamma] \eval v \triangleright k$ for some $v \in \Val$, $k \in L$
  \item $k \preceq b[\gamma]$
  \item $(v, b[\gamma]) \in \mathrm{ob}(\llbracket A \rrbracket(r, \gamma))$
  \item $\mathsf{cost}_A(r, \gamma)(v, b[\gamma]) = b[\gamma]$
\end{enumerate}
\end{theorem}

\begin{proof}
By induction on the typing derivation, as in Theorem~\ref{thm:cost-soundness}, but now tracking the groupoid structure. The key cases:

\paragraph{Case (J-eliminator).} Given $p : \IdType_A(a, b)$ interpreted as a morphism $\llbracket a \rrbracket \to \llbracket b \rrbracket$ in $\llbracket A \rrbracket(r, \gamma)$, and $d : C[a, \refl_a]$, we have $\J(p, d) : C[b, p]$ by transport along $p$. The cost is $b_p \oplus b_d \oplus \delta_\J$.

\paragraph{Case ($\Pi$-Elim).} As before, but now morphisms (identity proofs) between functions are preserved by application.
\end{proof}

\subsection{Resource-Bounded CwFs over Groupoids}

\begin{definition}[RB-CwF over groupoids]
\label{def:rb-cwf-gpd}
A \emph{resource-bounded category with families over groupoids} (RB-CwF-Gpd) over lattice $L$ consists of:
\begin{enumerate}
  \item A category $\mathcal{C}$ with a terminal object, where each hom-set $\mathcal{C}(\Gamma, \Delta)$ carries a groupoid structure (making $\mathcal{C}$ a $(2,1)$-category)
  \item A pseudofunctor $\mathsf{Ty} : \mathcal{C}^{op} \to \mathbf{Gpd}$ of types
  \item For each $\Gamma \in \mathcal{C}$ and $A \in \mathsf{Ty}(\Gamma)$, a groupoid $\mathsf{Tm}(\Gamma, A)$ of terms
  \item Comprehension with the universal property up to isomorphism
  \item An internal discrete lattice object $\mathbb{L}_\mathcal{C}$
  \item Cost functors $\mathsf{cost}_A : \mathsf{Tm}(\Gamma, A) \to \mathsf{disc}(\mathbb{L}_\mathcal{C})$
  \item Type formers ($\Pi$, $\Sigma$, $\mathsf{Id}$, $\Nat$, $\Vectt$, etc.) respecting the groupoid structure and costs
\end{enumerate}
\end{definition}
\begin{remark}[Relation to ordinary CwFs]
Our notion of RB-CwF-Gpd is a resource-graded refinement of the standard categories with families (CwFs) used to model Martin--L\"of type theory, such as those of Hofmann~\cite{hofmann1997syntax}.
When the internal lattice object is taken to be terminal and all cost functors $\mathsf{cost}_A$ are constant at the least element of $L$, an RB-CwF-Gpd reduces to an ordinary CwF over groupoids.
In this sense RB-CwF-Gpd is not a new semantic framework but an enrichment of the usual CwF structure with cost data.
\end{remark}

\begin{example}[Syntactic RB-CwF-Gpd]
The syntax of RB-MLTT forms an RB-CwF-Gpd where:
\begin{itemize}
  \item $\mathcal{C}(\Gamma, \Delta)$ = substitutions up to definitional equality
  \item Groupoid structure on $\mathsf{Tm}(\Gamma, A)$: terms identified up to definitional equality, with identity proofs as morphisms
  \item $\mathbb{L}_\mathcal{C}$ = discrete groupoid on $L$
\end{itemize}
\end{example}

\begin{example}[Groupoid presheaf RB-CwF-Gpd]
The groupoid presheaf model $\mathbf{Gpd}^L$ forms an RB-CwF-Gpd where:
\begin{itemize}
  \item $\mathcal{C}$ = groupoid-valued presheaves over $L$
  \item $\mathsf{Ty}(\llbracket \Gamma \rrbracket)$ = groupoid-valued presheaves over $\int \llbracket \Gamma \rrbracket$
  \item $\mathsf{Tm}$ = sections of the fibration
  \item $\mathbb{L}_\mathcal{C} = \mathbb{L}$ (discrete)
\end{itemize}
\end{example}

\subsection{Initiality}

\begin{theorem}[Initiality]
\label{thm:initiality}
The syntactic RB-CwF-Gpd is initial: for any RB-CwF-Gpd $\mathcal{M}$, there exists a unique (up to isomorphism) morphism of RB-CwF-Gpds $\llbracket - \rrbracket_\mathcal{M} : \mathbf{Syn} \to \mathcal{M}$.
\end{theorem}

\begin{proof}
By induction on syntax, now respecting the groupoid structure:

\paragraph{Contexts.} $\llbracket \Gamma \rrbracket_\mathcal{M}$ defined by comprehension.

\paragraph{Types.} $\llbracket A \rrbracket_\mathcal{M}$ defined using type formers in $\mathcal{M}$.

\paragraph{Terms.} $\llbracket t \rrbracket_\mathcal{M}$ defined using term formers. For identity proofs, use the groupoid structure.

\paragraph{Identity proofs.} For $p : a = b$ in context $\Gamma$, $\llbracket p \rrbracket_\mathcal{M}$ is a morphism $\llbracket a \rrbracket_\mathcal{M} \to \llbracket b \rrbracket_\mathcal{M}$ in the groupoid $\mathsf{Tm}(\llbracket \Gamma \rrbracket_\mathcal{M}, \llbracket A \rrbracket_\mathcal{M})$.

\paragraph{Cost preservation.} The cost functor commutes with interpretation:
\[
\mathsf{cost}_{\llbracket A \rrbracket_\mathcal{M}}(\llbracket t \rrbracket_\mathcal{M}) = \llbracket b \rrbracket_\mathcal{M}
\]
where $\Gamma \vdash_{r;\,b} t : A$.

\paragraph{Uniqueness.} Up to natural isomorphism, forced by the universal properties of type formers and the groupoid structure.
\end{proof}

\begin{corollary}[Sound embedding]
For any closed term $\cdot \vdash_{r;\,b} t : A$ and any RB-CwF-Gpd $\mathcal{M}$:
\begin{enumerate}
  \item $\llbracket t \rrbracket_\mathcal{M} \in \mathrm{ob}(\mathsf{Tm}(\mathbf{1}, \llbracket A \rrbracket_\mathcal{M}))$
  \item $\mathsf{cost}(\llbracket t \rrbracket_\mathcal{M}) \preceq b$
  \item Identity proofs are preserved: $\llbracket \refl_a \rrbracket_\mathcal{M} = \mathrm{id}_{\llbracket a \rrbracket_\mathcal{M}}$
\end{enumerate}
\end{corollary}

\begin{remark}[Relation to higher models]
The groupoid model captures the 1-truncated structure of types (sets with non-trivial identity). For full HoTT with univalence and higher inductive types, one needs $\infty$-groupoid-valued presheaves (simplicial sets or similar). This is the subject of RB-HoTT (future work).
\end{remark}

\section{Case Studies}
\label{sec:case-studies}

\subsection{Vector Sum: $O(n)$}

\begin{align*}
&\mathsf{sum} : \Pitype_{n:\Nat}^{\bot} \Vectt(\Nat, n) \to^{3n + 2} \Nat \\
&\mathsf{sum} = \lambda n.\, \lambda v.\, \vecrec(v, \zmark, \lambda m.\lambda a.\lambda w.\lambda ih.\, a + ih)
\end{align*}

\paragraph{Bound derivation.}
\begin{itemize}
  \item Base: returning $\zmark$ costs $\delta_\zmark = 1$
  \item Step: addition costs $\delta_+ = 2$, cons traversal costs $1$
  \item Total: $1 + n \cdot 3 = 3n + 1$, plus overhead $\delta_{\vecrec} = 1$
  \item Final bound: $3n + 2$
\end{itemize}

By Theorem~\ref{thm:cost-soundness}, $\mathsf{sum}\,n\,v$ evaluates with cost $\leq 3n + 2$.

\subsection{Vector Reverse: $O(n)$}

\begin{align*}
&\mathsf{reverse} : \Pitype_{n:\Nat}^{\bot} \Pitype_{A:\U}^{\bot} \Vectt(A, n) \to^{4n + 2} \Vectt(A, n) \\
&\mathsf{reverse} = \lambda n.\lambda A.\lambda v.\, \mathsf{rev\text{-}acc}(v, \mathsf{nil})
\end{align*}

where $\mathsf{rev\text{-}acc}$ is the tail-recursive helper with accumulator.

\subsection{Binary Search: $O(\log n)$}

For a sorted vector, binary search
\paragraph{Data-structure assumption.}
This bound assumes $O(1)$ random access (direct indexing) as in \S\ref{sec:safe-index}. If vectors are represented purely inductively and lookup traverses from the head, then each probe costs $O(n)$ and the overall bound becomes $O(n\log n)$.
 has logarithmic cost:

\begin{align*}
&\mathsf{bsearch} : \Pitype_{n:\Nat}^{\bot} \Vectt(\Nat, n) \to \Nat \to^{5 \cdot \lceil\log_2(n+1)\rceil + 3} \mathsf{Option}(\Fin(n))
\end{align*}

\paragraph{Bound derivation.}
\begin{itemize}
  \item Each iteration: comparison ($2$), index computation ($2$), slice ($1$) = $5$ total
  \item Iterations: $\lceil \log_2(n+1) \rceil$
  \item Base case and wrapping: $3$
  \item Total: $5 \cdot \lceil\log_2(n+1)\rceil + 3$
\end{itemize}

\subsection{Merge Sort: $O(n \log n)$}

\begin{align*}
&\mathsf{mergesort} : \Pitype_{n:\Nat}^{\bot} \Pitype_{A:\U}^{\bot} (A \to A \to^{c_{\mathsf{cmp}}} \mathsf{Bool}) \to \Vectt(A, n) \to^{c \cdot n \cdot \lceil\log_2(n+1)\rceil + d} \Vectt(A, n)
\end{align*}

\paragraph{Bound derivation.}
\begin{itemize}
  \item Divide: split vector in half, cost $O(n)$
  \item Conquer: two recursive calls on size $n/2$
  \item Combine: merge two sorted halves, cost $O(n)$
  \item Recurrence: $T(n) = 2T(n/2) + O(n)$
  \item Solution: $T(n) = O(n \log n)$
\end{itemize}

\subsection{Matrix Multiplication: $O(n^3)$}

For $n \times n$ matrices represented as $\Vectt(\Vectt(R, n), n)$:

\begin{align*}
&\mathsf{matmul} : \Pitype_{n:\Nat}^{\bot} \mathsf{Mat}(R, n) \to \mathsf{Mat}(R, n) \to^{c \cdot n^3 + d} \mathsf{Mat}(R, n)
\end{align*}

where $\mathsf{Mat}(R, n) = \Vectt(\Vectt(R, n), n)$.

\section{Implementation Notes}
\label{sec:implementation}

\subsection{Lean 4 Integration}

The framework can be implemented in Lean 4 using:

\begin{lstlisting}[language=Haskell]
-- Size-indexed vectors
inductive Vec (A : Type) : Nat -> Type where
  | nil  : Vec A 0
  | cons : A -> Vec A n -> Vec A (n + 1)

-- Bound annotation (custom attribute)
@[bound (fun n => 3 * n + 2)]
def sum : (n : Nat) -> Vec Nat n -> Nat
  | 0,     .nil        => 0
  | n + 1, .cons x xs  => x + sum n xs

-- Bound theorem
theorem sum_bound (n : Nat) (v : Vec Nat n) :
    cost (sum n v) <= 3 * n + 2 := by
  induction n with
  | zero => simp [sum, cost]
  | succ n ih => simp [sum, cost]; omega
\end{lstlisting}

\subsection{Agda Integration}

In Agda with sized types:

\begin{lstlisting}
data Vec (A : Set) : Nat -> Set where
  []  : Vec A zero
  _::_ : forall {n} -> A -> Vec A n -> Vec A (suc n)

sum : forall {n} -> Vec Nat n -> Nat
sum []       = 0
sum (x :: xs) = x + sum xs

-- Bound expressed in the type
sum-bounded : forall n -> (v : Vec Nat n) -> Cost (sum v) <= 3 * n + 2
\end{lstlisting}

A Lean 4 mechanization of the underlying (non resource-bounded) extrinsic MLTT
syntax and typing infrastructure used as the implementation substrate for this paper
is available in Corey Thuro's \texttt{RB-TT} repository. Concretely, the file\texttt{src/RBTT/Core/ExtrinsicMLTT.lean} provides the extrinsic judgemental
presentation (syntax, contexts, substitution, and typing predicates) that we
intend to extend with resource annotations and bound synthesis as described here:

\url{https://github.com/CoreyThuro/RB-TT/blob/main/src/RBTT/Core/ExtrinsicMLTT.lean}.

\section{Related Work}
\label{sec:related}

\paragraph{Dependent types and costs.} Danielsson~\cite{danielsson2008lightweight} develops lightweight semiformal complexity analysis using dependent types in Agda. Our work makes bounds first-class in the type system with formal soundness guarantees.

\paragraph{Sized types.} Abel~\cite{abel2010miniagda} and Hughes et al.~\cite{hughes1996sized} use size annotations for termination. We extend this idea to resource bounds, not just termination.

\paragraph{Type theory foundations.} Martin-L\"{o}f~\cite{martinlof1984} establishes the foundations; Nordstr\"{o}m et al.~\cite{nordstrom1990} provide an introduction. We extend MLTT with resource tracking.

\paragraph{Categorical semantics.} Hofmann~\cite{hofmann1997syntax} develops syntax and semantics of dependent types using CwFs. Jacobs~\cite{jacobs1999} provides comprehensive coverage. Our RB-CwF extends this with cost structure.

\paragraph{Coeffects and graded types.} Petricek et al.~\cite{petricek2014coeffects}, Brunel et al.~\cite{brunel2014core}, and Orchard et al.~\cite{orchard2019granule} develop graded type systems. We extend grading to dependent types with size-indexed bounds.

\paragraph{Quantitative type theory.} Atkey~\cite{atkey2018syntax} and McBride~\cite{mcbride2016got} track usage quantities. We track computational cost rather than usage.

\paragraph{RAML and amortized analysis.} Hoffmann et al.~\cite{hoffmann2017raml} infer polynomial bounds automatically. Our approach certifies bounds via typing rather than inference.

\paragraph{Prior work.} This paper extends \cite{mannucci2025rbtt} from simple types to dependent types.

\section{Limitations and Future Work}
\label{sec:future}

\paragraph{Current limitations.}
\begin{itemize}
  \item Bounds must be expressible in our bound language (polynomial, logarithmic)
  \item No automatic bound inference; bounds are checked, not synthesized
  \item Intensional identity types only; no univalence
  \item No higher inductive types
\end{itemize}

\paragraph{Future directions.}
\begin{itemize}
  \item \textbf{RB-HoTT:} Extend to Homotopy Type Theory with resource-bounded paths
  \item \textbf{Automatic inference:} Integrate with RAML-style inference
  \item \textbf{Cubical type theory:} Resource-bounded cubical structure
  \item \textbf{Probabilistic bounds:} Expected cost analysis
  \item \textbf{Mechanization:} Full formalization in Agda or Lean
  \item \textbf{Quotient types:} Resource-bounded quotients
\end{itemize}

\section{Conclusion}

We have extended resource-bounded type theory to Martin-L\"{o}f Type Theory, enabling size-indexed cost bounds for dependent types. The key innovations are:

\begin{enumerate}
  \item \textbf{Dependent bounds:} $\Pitype_{x:A}^{b(x)} B$ where $b(x)$ depends on the argument
  \item \textbf{Inductive families:} $\Vectt(A, n)$ and $\Fin(n)$ with size-dependent eliminator bounds
  \item \textbf{Resource universes:} $\U_r$ tracking type formation cost
  \item \textbf{Cost soundness:} Theorem~\ref{thm:cost-soundness} establishing bounds over-approximate costs
  \item \textbf{Semantic model:} Presheaves with dependent structure and initiality
\end{enumerate}

This bridges the gap between dependent type theory and quantitative resource analysis, enabling certified cost bounds for size-dependent algorithms expressed directly in types. The framework handles common complexity classes ($O(1)$, $O(n)$, $O(\log n)$, $O(n \log n)$, $O(n^2)$, $O(n^3)$) and validates bounds for fundamental algorithms including vector operations, binary search, and sorting.

Future work will extend the framework to Homotopy Type Theory (RB-HoTT), exploring resource-bounded paths and higher structure.


\begin{thebibliography}{99}

\bibitem{mannucci2025rbtt}
M.~A. Mannucci and C.~Thuro.
\newblock Resource-bounded type theory: Compositional cost analysis via graded modalities.
\newblock \emph{arXiv preprint}, 2025.

\bibitem{martinlof1984}
P.~Martin-L\"{o}f.
\newblock \emph{Intuitionistic Type Theory}.
\newblock Bibliopolis, Naples, 1984.

\bibitem{nordstrom1990}
B.~Nordstr\"{o}m, K.~Petersson, and J.~M. Smith.
\newblock \emph{Programming in Martin-L\"{o}f's Type Theory}.
\newblock Oxford University Press, 1990.

\bibitem{hofmann1997syntax}
M.~Hofmann.
\newblock Syntax and semantics of dependent types.
\newblock In \emph{Semantics and Logics of Computation}, pages 79--130. Cambridge University Press, 1997.

\bibitem{hofmann1998groupoid}
M.~Hofmann and T.~Streicher.
\newblock The groupoid interpretation of type theory.
\newblock In \emph{Twenty-Five Years of Constructive Type Theory}, pages 83--111. Oxford University Press, 1998.

\bibitem{jacobs1999}
B.~Jacobs.
\newblock \emph{Categorical Logic and Type Theory}.
\newblock Elsevier, 1999.

\bibitem{abel2010miniagda}
A.~Abel.
\newblock MiniAgda: Integrating sized and dependent types.
\newblock In \emph{Proc.\ PAR}, 2010.

\bibitem{hughes1996sized}
J.~Hughes, L.~Pareto, and A.~Sabry.
\newblock Proving the correctness of reactive systems using sized types.
\newblock In \emph{Proc.\ POPL}, pages 410--423. ACM, 1996.

\bibitem{danielsson2008lightweight}
N.~A. Danielsson.
\newblock Lightweight semiformal time complexity analysis for purely functional data structures.
\newblock In \emph{Proc.\ POPL}, pages 133--144. ACM, 2008.

\bibitem{petricek2014coeffects}
T.~Petricek, D.~Orchard, and A.~Mycroft.
\newblock Coeffects: A calculus of context-dependent computation.
\newblock In \emph{Proc.\ ICFP}, pages 123--135. ACM, 2014.

\bibitem{brunel2014core}
A.~Brunel, M.~Gaboardi, D.~Mazza, and S.~Zdancewic.
\newblock A core quantitative coeffect calculus.
\newblock In \emph{Proc.\ ESOP}, volume 8410 of \emph{LNCS}, pages 351--370. Springer, 2014.

\bibitem{orchard2019granule}
D.~Orchard, V.-B.~Liepelt, and H.~Eades~III.
\newblock Quantitative program reasoning with graded modal types.
\newblock \emph{Proc.\ ACM Program.\ Lang.}, 3(ICFP):110:1--110:30, 2019.

\bibitem{atkey2018syntax}
R.~Atkey.
\newblock Syntax and semantics of quantitative type theory.
\newblock In \emph{Proc.\ LICS}, pages 56--65. ACM, 2018.

\bibitem{mcbride2016got}
C.~McBride.
\newblock I got plenty o' nuttin'.
\newblock In \emph{A List of Successes That Can Change the World}, volume 9600 of \emph{LNCS}, pages 207--233. Springer, 2016.

\bibitem{hoffmann2017raml}
J.~Hoffmann, K.~Aehlig, and M.~Hofmann.
\newblock Multivariate amortized resource analysis.
\newblock \emph{ACM Trans.\ Program.\ Lang.\ Syst.}, 39(1):3:1--3:58, 2017.

\bibitem{maclane1992}
S.~Mac~Lane and I.~Moerdijk.
\newblock \emph{Sheaves in Geometry and Logic}.
\newblock Springer-Verlag, 1992.

\bibitem{awodey2018natural}
S.~Awodey.
\newblock Natural models of homotopy type theory.
\newblock \emph{Mathematical Structures in Computer Science}, 28(2):241--286, 2018.

\bibitem{dybjer1996internal}
P.~Dybjer.
\newblock Internal type theory.
\newblock In \emph{Types for Proofs and Programs}, volume 1158 of \emph{LNCS}, pages 120--134. Springer, 1996.

\end{thebibliography}
\end{document}